\documentclass[10pt]{amsart}
\usepackage{amssymb,mathrsfs,color}
\usepackage{pinlabel}

\usepackage{amsmath, amsfonts, amsthm, verbatim}
\usepackage{epstopdf, mathtools}
\mathtoolsset{showonlyrefs}
\usepackage{color}
\usepackage{esint}
\usepackage{stmaryrd}

\usepackage{empheq}
\usepackage[shortlabels]{enumitem}

\newcommand{\cl}{\mathcal}
\newcommand{\rar}{\rightarrow}

\newcommand{\bs}[1]{\boldsymbol{#1}}

\definecolor{deepgreen}{cmyk}{1,0,1,0.5}

\newcommand{\al}{\alpha}

\newcommand{\p}{\partial}

\numberwithin{equation}{section}

\newtheorem{thm}{Theorem}[section]

\newtheorem{prop}[thm]{Proposition}

\theoremstyle{remark}

\usepackage{tensor}

\usepackage{graphicx}
\usepackage{xcolor} 
\usepackage{tensor}
\usepackage{slashed}
\usepackage{cite}
\definecolor{green}{rgb}{0,0.8,0} 

\usepackage[bottom]{footmisc}

\newcommand{\tr}{\textrm{tr}}

\newcommand{\bbE}{\mathbb E}

\newcommand{\bbR}{\mathbb R}
\newcommand{\bbS}{\mathbb S}

\newcommand{\mfm}{\bs{\mathsf{m}}}
\newcommand{\mfn}{\bs{\mathsf{n}}}

\newcommand{\mfv}{\bs{\mathsf{v}}}

\newcommand{\mfu}{\bs{\mathsf{u}}}

\newcommand{\msm}{\bs{\mathsf{m}}}
\newcommand{\msn}{\bs{\mathsf{n}}}
\newcommand{\msv}{\bs{\mathsf{v}}}
\newcommand{\msa}{\bs{\mathsf{a}}}
\newcommand{\msb}{\bs{\mathsf{b}}}

\newcommand{\msu}{\bs{\mathsf{u}}}

\newcommand{\msE}{\bs{\mathsf{E}}}
\newcommand{\msR}{\bs{\mathsf{R}}}

\begin{document}

\title[Strain-limited special Cosserat rod]{On an elastic strain-limiting special Cosserat rod model}

\author{K. R. Rajagopal and C. Rodriguez}

\begin{abstract}
	Motivated by recent strain-limiting models for solids and biological fibers, we introduce the first intrinsic set of nonlinear constitutive relations, between the geometrically exact strains and the components of the contact force and contact couple, describing a uniform, hyperelastic, strain-limiting special Cosserat rod. After discussing some attractive features of the constitutive relations (orientation preservation, transverse symmetry, and monotonicity), we exhibit several explicit equilibrium states under either an isolated end thrust or an isolated end couple. In particular, certain equilibrium states exhibit \emph{Poynting} like effects, and we show that under mild assumptions on the material parameters, the model predicts an explicit tensile \emph{shearing bifurcation}: a straight rod under a large enough tensile end thrust parallel to its center line can shear.  
\end{abstract}

\maketitle

\section{Introduction}

\subsection{Elastic rods}
{
A rod is a slender body that can be modeled by a deformable one-dimensional continuum, appropriately parameterized, wherein we can account for the elongation, shearing, bending, and twisting of the slender body. If the slender body in question is elastic, then we have a model for an elastic rod. A succinct history of the development of a theory for elastic rods from that for a three-dimensional elastic body can be found in pages 663--664 of the treatise \cite{Antman1973}. 

The special Cosserat theory provides an alternative development of a model to describe rods. In the static setting of this alternative approach, the configuration of the deformed slender body is modeled by a one-dimensional curve in space, 
\begin{align}
	[0,L] \ni s \mapsto \bs r(s) \in \bbE^3, 
\end{align}
(the \emph{center line}) to which a right-handed collection of orthonormal vector fields $ \{ \bs d_k(\cdot) \}_{k = 1}^3$ (the  \emph{directors}) is attached.\footnote{Such an approach was introduced by Duhem in \cite{Duhem1893} and E. and F. Cosserat in \cite{Cosserats1907, Cosserats1909}.}  The directors $\{\bs d_1(s), \bs d_2(s)\}$ are viewed as tangent to the material cross section transversal to the center line at $\bs r(s)$; see Figure \ref{f:fig1} in Section 2.1.  The directors  
are able to deform independently of the center line, and the Darboux vector field $\bs u(\cdot)$ along the center line characterizes their deformation: 
\begin{align}
	\frac{d}{ds} \bs d_k(s) = \bs u(s) \times \bs d_k(s), \quad s \in [0,L], k = 1, 2, 3.
\end{align}  
The six components of $\bs u(s)$ and $\bs v(s) = \frac{d}{ds} \bs r(s)$ in the frame $\{\bs d_k(s)\}_{k = 1}^3$ are the measures of geometric strain in the theory. Two other vector fields along the rod, the contact force $\bs n(\cdot)$ and contact couple $\bs m(\cdot)$, are postulated to model how material segments exert forces and couples on other material segments. The differential equations expressing local balance of linear and angular momentum and the specification of constitutive relations between the geometric strains and the components of the rod's contact force and contact couple yield a closed system of equations capable of prediction. See Section 2 for a review of the elements of the special Cosserat theory needed for this work.  
An up-to-date discussion of the results pertinent to general elastic rod theories can be found in the exhaustive book \cite{AntmanBook} on nonlinear elasticity. 

\subsection{Motivation for the model}
In the case of the classical theory for rods developed within the context of 3-dimensional bodies, one usually assumes that the body in question belongs to the class of simple materials (see \cite{Noll58}), which within the confines of a purely mechanical theory leads to the stress being completely determined by the history of deformation. Elastic bodies are special cases of such simple materials. In \cite{Raj_Implicit03} the first author generalized the notion of simple materials to include implicit relations between the history of the stress and the history of the deformation gradient (see also the work \cite{Prusa2012}). A special subclass of such materials are elastic bodies that are related through an implicit relation between the Cauchy stress $\bs T$ and the deformations gradient $\bs F$, through
\begin{align}
\bs f(\bs T, \bs F) = \bs 0. \label{eq:1}
\end{align}
We notice that the class of classical Cauchy elastic bodies, satisfying $\bs T = \bs g(\bs F)$, are a special sub-class of the constitutive relations \eqref{eq:1}.                                                                                     
If the body is anisotropic, then \eqref{eq:1} reduces to (see \cite{RAJAGOPAL2015})
\begin{align}
\bs f(\bs R^* \bs T \bs R, \bs C) = \bs 0,        \label{eq:2}                      
\end{align}                                 
where $\bs F = \bs R \bs U$ is the polar decomposition of the deformation gradient and $\bs C = \bs F^* \bs F$ is the right Cauchy-Green tensor.
If the three-dimensional body has one preferential direction $\bs a$ with respect to its response, then the implicit equation will take the form
\begin{align} 
\bs f(\bs T, \bs F, \bs a) = \bs 0.
	\label{eq:4}
\end{align}
A special sub-class of \eqref{eq:2} are constitutive relations of the form
\begin{align}
\bs C = \bs h(\bs R^* \bs T \bs R),      \label{eq:5}
\end{align}
and if the body is isotropic, the constitutive relation takes the form
\begin{align}
	\bs B = \bs h(\bs T), \label{eq:6}
\end{align}
where $\bs B = \bs F \bs F^*$ is the left Cauchy-Green tensor. 
Our reason for documenting the constitutive relations given by \eqref{eq:5} and \eqref{eq:6} is to motivate the constitutive relations that we propose for our special Cosserat rod in this work.
However, we point out that the constitutive relations between the special Cosserat rod's geometric strains and the rod's contact force and couple that we propose \emph{are not derived} from a fully 3-dimensional relation like \eqref{eq:5} or \eqref{eq:6}. 
Instead, we propose \emph{intrinsic} constitutive relations between the geometric strains and the components of the rod's contact force and couple (see equations \eqref{eq:uconst}, \eqref{eq:u3const}, \eqref{eq:vconst} and \eqref{eq:v3const}). 

While it is customary and mathematically desirable to provide expressions for the components of the rod's contact force and contact couple in terms of the strains, we provide expressions for the strains in terms of the components of the rod's contact force and contact couple.
There is a sound reason for our choice for specifying the constitutive relation as we do. A central notion in Newtonian mechanics, and this includes continuum mechanics, is the notion of causality. From the standpoint of causality, applied forces and applied moments are the causes and the strains are the effect, and thus it would be appropriate to express the effects in terms of the causes than vice-versa (see \cite{RajConspectus}) for a discussion of the relevant issues).

The constitutive relations that we propose in this paper are motivated by constitutive relations introduced in \cite{Raj2010, RajSmallStrain, RajConspectus} for the response of three-dimensional elastic bodies wherein the left Cauchy-Green tensor depends nonlinearly on the stress. While the constitutive relations considered in \cite{Raj2010} is for a 3-dimensional elastic isotropic solid, we appropriately modify it to describe an elastic special Cosserat rod that is \emph{transversely hemitropic} (as defined in \cite{Healey02}). 
 
The constitutive relations introduced in this work have a very special feature, namely that they are strain-limiting, that is as the applied forces and couples increase, the strains increase monotonically and reach a finite limit asymptotically (see Figure \ref{f:fig2}). This feature of the constitutive relations, we feel, is a critical feature that would make it an appropriate model to describe many rod-like materials whose final tangent stiffness greatly exceeds the initial tangent stiffness, so much so that such materials are best approximated as materials with limiting stretch (see Section 5 for a discussion).  
More generally, over the past decade, works on strain-limiting constitutive relations for elastic solid bodies have investigated:
\begin{itemize}
	\item specific problems of simple shear, torsion, extension etc. \cite{Bustamante2009, BustRaj2010, BustRaj2011, RajSmallStrain}, 
\item time independent and dependent boundary value problems \cite{BulicekMalekRajSuli2014, BulMalSuli2015, BulMalekRajWal_Existence15, GelmettiSuli2019, BULPATSULISENGUL21, BULPATSULISENGUL22},
\item the modeling of gum metal \cite{RajSmallStrain, DevSanKannRaj2017, KulvaitMalRaj2017}, 
\item the theory of fracture \cite{RajWalt2011, KulMalRaj2013, GuoMallRajWalt2015, ZappBertoRaj2016, ItouKovRaj2017, ItouKovRaj2017b, KulMalRaj2019},
\item materials with density dependent moduli and damage \cite{RAJDENSITY21, D32021, D32021II, MURRURAJ21UNI, MURRAJRAJ21BI, ITOUKOVRAJ21, PRUSARAJWINE22, VAJMURRAJ22}, 
\item viscoelasticity \cite{ERBAYSENG15, ERBAYSENG20, ERBAYetal2020 ,SENGUL2021},
\item wave propagation \cite{KannanRajSacc2014, RajSacc2014, BustSfy2015, MagMasHar2018, MenOreBust2018},
\item elastic strings \cite{Rodriguez21, Rodriguez22}.
\end{itemize}
However, this paper is the \emph{first work} studying a set of strain-limiting constitutive relations for special Cosserat rods.

We comment that it is not only causality that influences our choice of the constitutive relations for the strains in terms of the contact force and couple. In a real material, the limit for a measure of strain is (essentially) reached for a finite applied force, and the body seems to behave elastically (to first approximation) for a slightly larger force. But, upon applying a much larger force, the material behaves inelastically and cannot be modeled within a purely elastic framework.  In the regime of elastic response, such a material could be described by a constitutive relation wherein the strain is a function of the stress and not the stress as a function of the strain. That is, the strain versus stress relation is not invertible. See Figure \ref{f:fig2}.  Such constitutive relations are not Cauchy elastic and hence not Green elastic. Thus, such bodies must be described either by relations such as \eqref{eq:5} or even the more general class \eqref{eq:2}, with these relations not being invertible. Recently, the second author studied such stretch-limiting constitutive relations for elastic strings in both static \cite{Rodriguez21} and dynamic \cite{Rodriguez22} settings, and the authors studied certain inelastic behavior for the quasistatic motion of an inextensible, unshearable, viscoelastic rod \cite{RajRod22}.  

\begin{figure}[b]
	\centering
	\includegraphics[width=\textwidth]{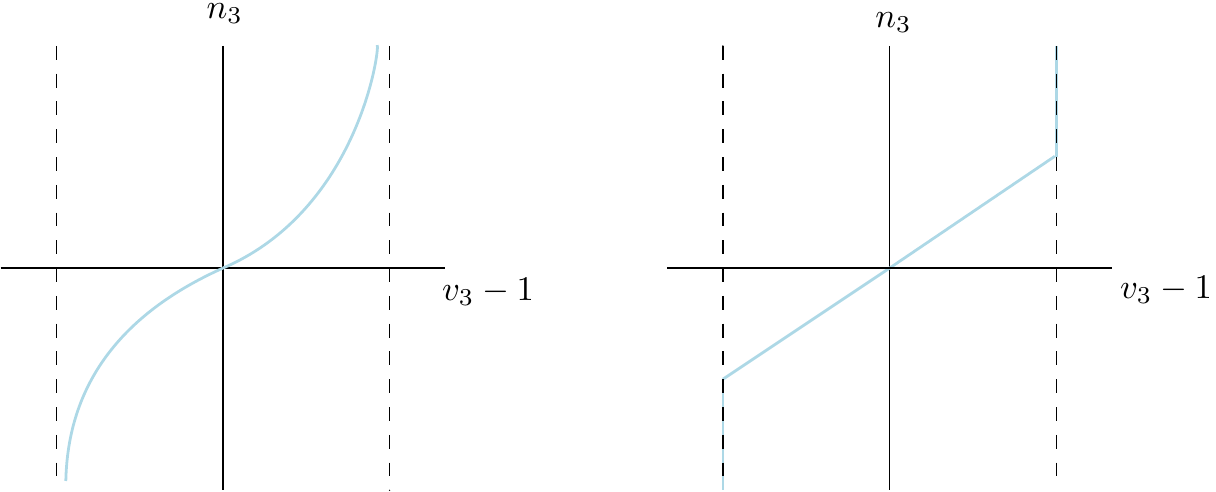}
	\caption{Two types of strain-limiting constitutive relationships between the dilatation strain $v_3$ and tension $n_3$. The first relation is qualitatively like those we consider in this work, and it can be expressed as $n_3 = \hat n_3(v_3 - 1)$ while the latter cannot.}
	\label{f:fig2}
\end{figure}

\subsection{Main results and outline}

An outline of the main results and the organization of our study are as follows. In Section 2 we briefly review the elements of the static special Cosserat theory used in this paper. In Section 3, we introduce our intrinsic set of strain-limiting constitutive relations between the geometrically exact strains and the components of the contact force and contact couple describing a uniform, special Cosserat rod capable of bending, twisting, shearing and stretching. See equations \eqref{eq:uconst}, \eqref{eq:u3const}, \eqref{eq:vconst} and \eqref{eq:v3const}.  We show that these relations arise from a complementary energy depending on the rod's contact force and contact couple. 
These relations can easily be specialized to enforce unshearability constraints, 
\begin{align}
	\bs r'(s) = v_3(s) \bs d_3(s), \quad s \in [0,L],
\end{align}
or unshearability and inextensibility constraints, 
\begin{align}
	\bs r'(s) = \bs d_3(s), \quad s \in [0,L],
\end{align}
but we do not pursue this here. We then show that these relations can be inverted, expressing the components of the rod's contact couple and contact force in terms of the strains, and that these relations can be derived from a stored energy. All but one of the parameters defining these relations readily admit physical interpretations in terms of small-strain material moduli (but we emphasize that our relations are in terms of large exact strains). See Section 3.1. In the remainder of Section 3, we establish the following attractive properties of the constitutive relations, valid for all values of contact force and contact couple: 
\begin{itemize}
\item orientation preservation\footnote{By \emph{orientation} we mean either a choice of direction for the director $\bs d_3$, relative to the center line, or a stronger condition in terms of the local invertibility of an associated constrained 3-dimensional deformation. See Section 3.2 for more.} (Section 3.2), 
\item transverse hemitropy (Section 3.3), 
\item monotonicity of the constitutive relations (Section 3.4). 
\end{itemize}  
To the authors' knowledge we are unaware of any other explicit shearable, extensible, special Cosserat rod model satisfying all three of these properties for all values of contact force and contact couple. In Section 4, we exhibit several explicit equilibrium states, under a prescribed end thrust, subject to either isolated contact forces (Section 4.2) or isolated contact couples (Section 4.3). In the former case, the center line of the rod is straight, the end thrust is parallel to the center line, and we find that when a tensile end thrust for the rod exceeds a critical value, we have a bifurcation that can be associated with spontaneous shearing of the rod (see Figure \ref{f:fig3}). In the case of isolated contact couples, our constitutive relations also predict an interesting feature, namely if the chirality of the rod is opposite the chirality of the applied couple, then the applied couple elongates the rod: a Poynting type effect.\footnote{The Poynting effect was first observed by Wertheim in the 1850's; see \cite{Wertheim1857} and the extended discussion in \cite{MechanicsSolidsI}.} Finally, we exhibit an explicit 2-parameter family of equilibrium states with helical center lines subject to isolated contact couples. In general, our discussion in Section 4 is from the semi-inverse standpoint: certain aspects of the equilibrium state will be fixed at the start in order to solve for the kinematic variables and obtain simple, explicit solutions. Finally, in Section 5 we discuss potential applications for the model and future work.

}
\section{Preliminaries}

In this section we briefly review the theory of special Cosserat rods (see Ch. 8 of \cite{AntmanBook} for a more comprehensive introduction).

\subsection{Kinematics and strains}
Let $\bbE^3$ be $3$-dimensional Euclidean space, and let $\{\bs g_k \}$ be a fixed right-handed orthonormal basis for the vector space $\bbR^3$. Let $[0,L]$ be the reference interval parameterizing the material points of a uniform rod, with reference length $L$, that is straight in the reference configuration. The (deformed) \emph{configuration} of a special Cosserat rod is defined by a triple: 
\begin{align}
[0,L] \ni s \mapsto (\bs r(s), \bs d_1(s), \bs d_2(s)) \in \bbE^3 \times \bbR^3 \times \bbR^3,
\end{align}
with $\bs d_1(s)$ and $\bs d_2(s)$ orthonormal for each $s$. The curve $\bs r(\cdot)$ is the \emph{center line} of the rod, and $\{\bs d_1(s), \bs d_2(s)\}$ are the \emph{directors} at $s$, vectors regarded as tangent to the material cross section transversal to the center line at $\bs r(s)$. Let
\begin{align*}
	\bs d_3(s) = \bs d_1(s) \times \bs d_2(s).
\end{align*}
Then $\{ \bs d_k(s) \}$ is a right-handed orthonormal basis for $\bbR^3$ for each $s$, and it describes the configuration of the deformed material cross section at $\bs r(s)$ (see Figure \ref{f:fig1}). 

\begin{figure}[b]
	\centering
\includegraphics[width=\textwidth]{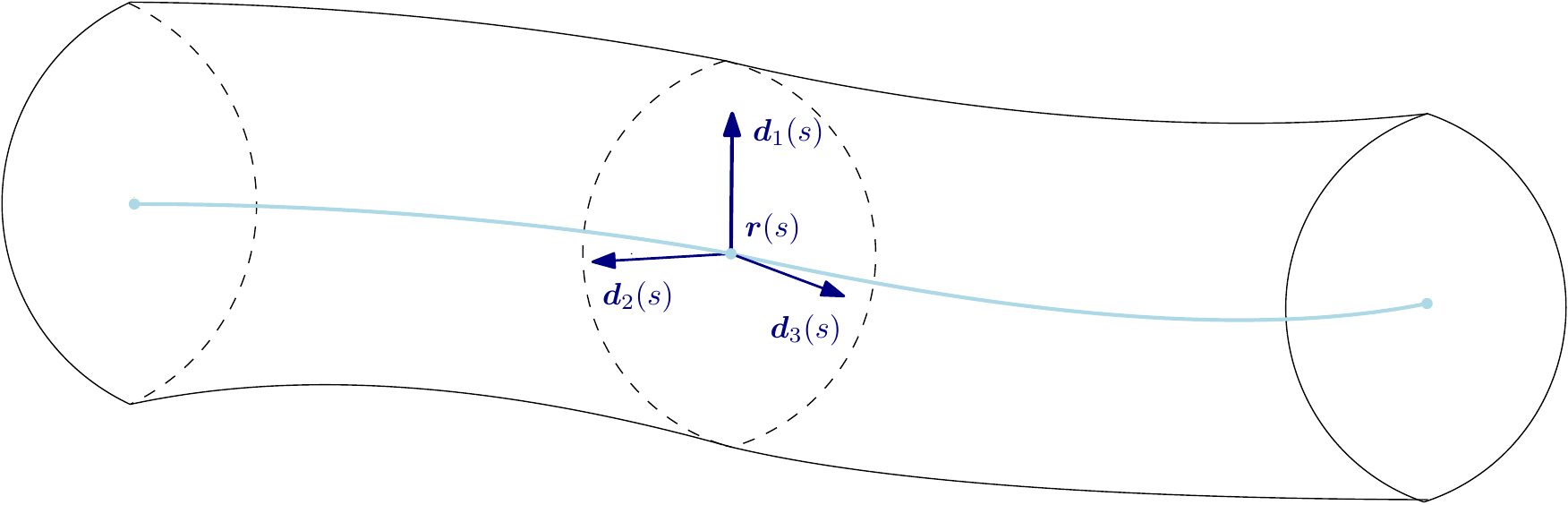}
\caption{The kinematic variables defining the configuration of a special Cosserat rod.}
\label{f:fig1}
\end{figure}

Since, for each $s$, $\{ \bs d_k(s) \}$ is a positively oriented orthonormal basis for $\bbR^3$, there exists a unique vector field (the \emph{Darboux vector field})
\begin{align*}
	\bs u(s) = u_k(s) \bs d_k(s) \in \bbR^3, \quad s \in [0,L],
\end{align*} 
such that 
\begin{align}
  \bs d_k'(s) = \bs u(s) \times \bs d_k(s),
\end{align}
where $'\mbox{ } = \frac{d}{ds}$.
In fact, one can solve for $u(s)$ explicitly, 
\begin{align}
	\bs u(s) = \frac{1}{2} \bs d_k(s) \times \bs d_k'(s). \label{eq:uequation}
\end{align}
The components $u_1$ and $u_2$ are referred to as the \emph{flexural strains}, and the component $u_3$ is referred to as the \emph{torsional strain} (or \emph{twist}).  
We may also express the tangent vector to the center line at $s$ in the basis $\{ \bs d_k(s) \}$ via 
\begin{align}
	\bs r'(s) = v_k(s) \bs d_k(s).
\end{align}
The components $v_1$ and $v_2$ are referred to as the \emph{shear strains}. The component $v_3$ is referred to as the \emph{dilatation strain}, and an orientation of the director $\bs d_3(s)$ relative to the center line is fixed by requiring that for all $s \in [0,L]$, 
\begin{align}
 v_3(s) > 0. \label{eq:v3pos}
\end{align} 
The restriction \eqref{eq:v3pos} also implies that the \emph{stretch} of the rod is never zero, $|\bs r'| > 0$, and that the rod cannot be sheared so severely that a section becomes tangent to the center line.

\subsection{Balance laws}

Let $[a,b] \subseteq [0,L]$. We denote the contact force by $\bs n(s)$ so that the resultant force on the material segment $[a,b]$ be $[0,a) \cup (b,L]$ is given by 
\begin{align}
	\bs n(b) - \bs n(a).
\end{align}
The contact couple is denoted by $\bs m(s)$ so that the resultant contact couple about $\bs o \in \bbE^3$ on the material segment $[a,b]$ by $[0,a) \cup (b,L]$ is given by 
\begin{align}
	\bs m(b) + (\bs r(b) - \bs o) \times \bs n(b) - \bs m(a) - (\bs r(a) - \bs o) \times \bs n(a). 
\end{align}
If $\bs f(s)$ is an external body force per unit reference length and $\bs l(s)$  is an external body couple per unit reference length, then the classical equilibrium equations expressing balance of linear momentum and angular momentum are given by: 
\begin{align}
	\bs n '(s) + \bs f(s) &= \bs 0, \label{eq:neq} \\
	\bs m '(s) + \bs r '(s) \times \bs n(s) + \bs l(s) &= \bs 0. \label{eq:meq}
\end{align}
At each $s$, the contact force and contact couple may be expressed in the basis $\{\bs d_k(s)\}$ via 
\begin{align}
	\bs n(s) = n_k(s)\bs d_k(s), \quad \bs m(s) = m_k(s) \bs d_k(s). 
\end{align}
The components $n_1$ and $n_2$ are referred to as the \emph{shear forces}, and the component $n_3$ is referred to as the \emph{tension} (or \emph{axial force}). The components $m_1$ and $m_2$ are referred to as the \emph{bending couples} (or \emph{bending moments}), and the component $m_3$ is referred to as the \emph{twisting couple} (or \emph{twisting moment}). 

\section{Strain-limited rods}

In this section we introduce our special Cosserat rod model and discuss some of its fundamental properties. {
We are interested in developing a fully nonlinear theory for the response of the rod. Thus, the strains that we consider below are the geometrically exact strains and are not constrained to be small. Also, as mentioned in the introduction, in keeping with causality we prescribe the geometric strains in terms of the rod's contact force and contact couple, and we show that these relations can be derived from a complementary energy that depends on the contact force and contact couple.  This is analogous to the development of constitutive relations wherein the strain is expressed in terms of a Gibbs potential in the fully three-dimensional theory of continuum mechanics.  }

\subsection{Constitutive relations}

To specify the mechanical properties of the rod and close the equations \eqref{eq:neq} and \eqref{eq:meq}, we must posit relations between the components of the strains
\begin{align}
	\mfu = \begin{bmatrix}
		u_1 \\
		u_2 \\
		u_3
	\end{bmatrix}, \quad 
\mfv = \begin{bmatrix}
	v_1 \\
	v_2 \\
	v_3
\end{bmatrix},
\end{align}
and the components of the contact couple and force
\begin{align}
	\mfm = \begin{bmatrix}
		m_1 \\
		m_2 \\
		m_3
	\end{bmatrix}, \quad 
	\mfn = \begin{bmatrix}
		n_1 \\
		n_2 \\
		n_3
	\end{bmatrix}.
\end{align}

The relations we propose are inspired by a class of elastic strain-limiting models introduced by Rajagopal within the context of $3$-dimensional solid mechanics \cite{Raj2010}: 
\begin{align}
	\bs e = \Bigl (
	a^{-p} + | b \bs T|^{p}
	\Bigr )^{-1/p} \bs T, \label{eq:TB}
\end{align}
where $a, b, p > 0$, $\bs B$ is the left Cauchy-Green tensor, $\bs e = \frac{1}{2}(\bs I - \bs B^{-1})$ is the Almansi-Hamel strain tensor, $\bs T$ is the Cauchy stress tensor and $|\bs T|^2 = \tr(\bs T^2)$. We emphasize here that the relations that we propose are between the contact force, contact couple and \emph{geometrically exact, large strains} rather than linearized strains. To describe our special Cosserat rod model, let $p, \al, \beta, \gamma, \zeta, \eta \in (0,\infty)$ and $\iota \in \bbR$ with 
\begin{align}
	\al^2 \beta^2 - \iota^2 > 0. 
\end{align}
Then the quadratic forms
\begin{align}
	Q(\mfu, \mfv) &= \al^2 (u_1^2 + u_2^2) + \beta^2 u_3^2 + +\zeta^2 (v_1^2 + v_2^2)\\ &\quad+ \eta^2 (v_3 - 1)^2 + 2 \iota u_3(v_3 - 1), \label{eq:quadraticform} \\
	Q^*(\mfm, \mfn) &= \frac{1}{\al^2} (m_1^2 + m_2^2) + \frac{1}{\zeta^2} (v_1^2 + v_2^2) +  
	\frac{\eta^2}{\beta^2 \eta^2 - \iota^2} m_3^2 \\ &\quad + 
	\frac{\beta^2}{\beta^2 \eta^2 - \iota^2} n_3^2
	- \frac{2 \iota}{\beta^2 \eta^2 - \iota^2} m_3 n_3.  
\end{align}
are positive definite. We posit that the relations between the strains and the components of the contact couple and force are given by
\begin{align}
	u_\mu & = \Bigl (
	\gamma^p + Q^*(\mfm, \mfn)^{p/2}
	\Bigr )^{-1/p} \frac{1}{\al^2} m_\mu, \label{eq:uconst} \\
	u_3 & =  \Bigl (
	\gamma^p + Q^*(\mfm, \mfn)^{p/2}
	\Bigr )^{-1/p} \frac{1}{\beta^2 \eta^2 - \iota^2} (\eta^2 m_3 - \iota n_3 ), \label{eq:u3const} \\
	v_\mu & = \Bigl (
	\gamma^p + Q^*(\mfm, \mfn)^{p/2}
	\Bigr )^{-1/p} \frac{1}{\zeta^2} n_\mu, \label{eq:vconst} \\
	v_3-1 &= \Bigl (
	\gamma^p + Q^*(\mfm, \mfn)^{p/2}
	\Bigr )^{-1/p} \frac{1}{\beta^2 \eta^2 - \iota^2} (-\iota m_3 + \beta^2 n_3), \label{eq:v3const}
\end{align}
where Greek letters range over $\{1,2\}$.
We note that the model is strain-limiting in the sense that $\mfu, \mfv \in \bbR^3$ given by \eqref{eq:uconst}, \eqref{eq:u3const}, \eqref{eq:vconst} and \eqref{eq:v3const} satisfy 
\begin{align}
	Q(\mfu, \mfv) < 1, \label{eq:strainlimiting}
\end{align}
and thus, 
\begin{gather}
	(u_1^2 + u_2^2)^{1/2} < \frac{1}{\al}, \label{eq:u1u2bd} \\
	|u_3| < \eta (\beta^2 \eta^2 - \iota^2)^{-1/2}, \label{eq:u3bd} \\
	(v_1^2 + v_2^2)^{1/2} < \frac{1}{\zeta}, \label{eq:v1v2bd} \\
	|v_3 - 1| < \beta (\beta^2 \eta^2 - \iota^2)^{-1/2}. \label{eq:v3bd}  
\end{gather}
The constitutive relations \eqref{eq:uconst}, \eqref{eq:u3const}, \eqref{eq:vconst} and \eqref{eq:v3const} are derivable from a complementary energy $W^*(\mfm, \mfn)$, 
\begin{align}
	\begin{bmatrix}
		u_1 \\
		u_2 \\
		u_3 
	\end{bmatrix}	
		 = \frac{\p W^*}{\p \mfm}(\mfm, \mfn), \quad \begin{bmatrix}
		v_1 \\
		v_2 \\
		v_3 - 1
	\end{bmatrix} = \frac{\p W^*}{\p \mfn}(\mfm, \mfn),
\end{align}
where 
\begin{gather}
	W^*(\mfm, \mfn)
	= \frac{1}{2}\int_0^{Q^*(\mfm, \mfn)} (\gamma^p + t^{p/2})^{-1/p} dt. \label{eq:compenergy} 
\end{gather}

We observe that the relations \eqref{eq:uconst}, \eqref{eq:u3const}, \eqref{eq:vconst} and \eqref{eq:v3const} can be inverted to obtain the contact couple and contact force as functions of the strains, 
\begin{align}
	m_\mu &= \frac{\p W}{\p u_\mu} = \gamma \Bigl (
	1 - Q(\mfu, \mfv)^{p/2}
	\Bigr )^{-1/p} \alpha^2 u_\mu, \label{eq:mconst}\\ 
	m_3 &= \frac{\p W}{\p u_3} = \gamma \Bigl (
	1 - Q(\mfu, \mfv)^{p/2}
	\Bigr )^{-1/p} (\beta^2 u_3 + \iota (v_3 - 1) ), \label{eq:m3const} \\
	n_\mu &= \frac{\p W}{\p v_\mu} = \gamma \Bigl (
	1 - Q(\mfu, \mfv)^{p/2}
	\Bigr )^{-1/p} \zeta^2 v_\mu, \label{eq:nconst}\\
	n_3 &= \frac{\p W}{\p v_3} = \gamma \Bigl (
	1 - Q(\mfu, \mfv)^{p/2}
	\Bigr )^{-1/p} (\iota u_3 + \eta^2(v_3 - 1)), \label{eq:n3const}
\end{align}
for $\mfu, \mfv \in \bbR^3$ satisfying $Q(\mfu, \mfv) < 1$. {
The fact that the constitutive relations can be inverted is a consequence of the type of strain-limiting behavior \eqref{eq:uconst}, \eqref{eq:u3const}, \eqref{eq:vconst} and \eqref{eq:v3const} exhibit: the strains increase to their limiting values \emph{asymptotically} (see Section 3.4 for monotonicity). If, instead, the strain-limiting behavior was of the type wherein the strains reached a finite limiting value for all sufficiently large finite values of $|\mfn|$ and $|\mfm|$, then the relations between $(\mfu, \mfv)$ and $(\mfm, \mfn)$ would not be invertible (see Figure \ref{f:fig2}). 
} We now observe that our model is \emph{hyperelastic}: the relations \eqref{eq:mconst}, \eqref{eq:m3const}, \eqref{eq:nconst} and \eqref{eq:n3const} are derivable from a stored energy $W(\mfu, \mfv)$, 
\begin{align}
	\begin{bmatrix}
		m_1 \\
		m_2 \\
		m_3 
	\end{bmatrix}  = \frac{\p W}{\p \mfu}(\mfu, \mfv), \quad
	\begin{bmatrix}
	n_1 \\
	n_2 \\
	n_3 
\end{bmatrix}  = \frac{\p W}{\p \mfv}(\mfu, \mfv) 
\end{align}
where 
\begin{align}
	W(\mfu, \mfv)
= \frac{\gamma}{2}\int_0^{Q(\mfu, \mfv)} (1 - t^{p/2})^{-1/p} dt. \label{eq:stenergy}
\end{align}

The material constants $\alpha, \beta$ and $\iota$ scale like length, $\gamma$ scales like force, and $\zeta$ and $\eta$ are dimensionless. In particular, our model is parameterized by seven dimensionless parameters: $\gamma/F$, $\al/L, \beta/L, \iota/L$, $\zeta$, $\eta$ and $p$, where $F$ is the force scale. All but the constant $p$ admit the following physical interpretation. We note that the stress-strain relations  \eqref{eq:mconst}, \eqref{eq:m3const}, \eqref{eq:nconst} and \eqref{eq:n3const} are, to leading order in $\mfu$ and $\mfv - [0 \,\, 0 \,\, 1]^*$, given by
\begin{align}
	m_\mu &= \gamma \al^2 u_\mu, \\
	m_3 &= \gamma \beta^2 u_3 + \gamma \iota (v_3 - 1), \\
	n_\mu &= \gamma \zeta^2 n_\mu, \\
	n_3 &= \gamma \iota u_3 + \gamma \eta^2(v_3 - 1).  
\end{align}
Thus, the four constants $\gamma \al^2, \gamma \beta^2,$ $\gamma \zeta^2$ and $\gamma \eta^2$ can be interpreted as the rod's small-strain \emph{bending, twisting, shearing} and \emph{dilatational} material moduli. The constant $\gamma \iota$ couples twisting to extension and specifies the chirality of the rod.

\subsection{Orientation preservation}

We recall that a condition on the strain $v_3$ that fixes an orientation and guarantees that the stretch of the rod is never zero, $|\bs r'| > 0$, is that for all $s \in [0,L]$,
\begin{align}
	v_3(s) > 0. \label{eq:v3sign}
\end{align} 
If the material constants characterizing our rod model satisfy 
\begin{align}
	\beta(\beta^2 \eta^2 - \iota^2)^{-1/2} < 1 \iff 1 + \frac{\iota^2}{\beta^2} < \eta^2, \label{eq:matconst1}
\end{align}
then by \eqref{eq:v3bd} an arbitrary configuration of the rod satisfies \eqref{eq:v3sign}.  

Although mathematically 1-dimensional manifolds, special Cosserat rods are used to model slender 3-dimensional bodies. The condition \eqref{eq:v3pos} is a mild condition ensuring orientation preservation of the 1-dimensional object. However, as discussed by Antman in Section 4 of \cite{ANTMAN76}, consideration of the constrained deformation of the associated slender 3-dimensional body suggests requiring a much stronger condition for orientation preservation, described as follows. 

For simplicity of the ensuing discussion, suppose that for each $s$, the material cross section $\cl A(s)$ located at the material point $s$ of the slender body being modeled is circular  
\begin{align}
	\cl A(s)  = \Bigl \{ (x^1, x^2) \mid 
	(x^1)^2 + (x^2)^2\leq a^2
	\Bigr \}.
\end{align} 
We denote the reference configuration of 
 the slender $3$-dimensional body by
\begin{align} 
	\cl B = \{ (x^1, x^2, s) \mid s \in [0,L], (x^1,x^2) \in \cl A(s) \}.
\end{align}
Then the constrained deformation of the slender body
$\bs \chi : \cl B \rar \bbE^3$ determined by an arbitrary configuration $\bs r, \bs d_1, \bs d_2$ for a special Cosserat rod via 
\begin{align}
	\bs \chi (x^1, x^2, s) = \bs r(s) + x^1 \bs d_1(s) + x^2 \bs d_2(s), 
\end{align}
preserves orientation in the sense that
\begin{align}
	\det \nabla \bs \chi > 0 \mbox{ on } \cl B,
\end{align}
if and only if 
\begin{align}
	v_3(s) > 
	a ((u_1(s))^2 + (u_2(s))^2)^{1/2}.  \label{eq:v3strongpos}
\end{align}

We now show that under a mild constraint on the cross sectional radii and the parameters, a slender body modeled by our special Cosserat rod model satisfies \eqref{eq:v3strongpos}. Indeed, if  
\begin{align}
	a < \alpha (1 - \beta (\beta^2 \eta^2 - \iota^2)^{-1/2}), \label{eq:matconst2}
\end{align} 
 then by \eqref{eq:u1u2bd} and \eqref{eq:v3bd} the strains associated to an arbitrary configuration of our strain-limiting rod satisfy  
\begin{align}
	a ((u_1(s))^2 + (u_2(s))^2)^{1/2}
	\leq \frac{a}{\alpha} < 1 - \beta (\beta^2 \eta^2 - \iota^2)^{-1/2} < v_3(s).
\end{align}

\subsection{Transverse symmetry}

We now discuss the transverse symmetry of our model. Let
\begin{align}
\msE = \begin{bmatrix}
	1 & 0 & 0 \\
	0 & -1 & 0 \\
	0 & 0 & 1
\end{bmatrix}, \quad 
	\msR_\psi = \begin{bmatrix}
		\cos \psi & \sin \psi & 0 \\
		-\sin \psi & \cos \psi & 0 \\
		0 & 0 & 1
	\end{bmatrix}, \quad \psi \in [0,2\pi), 
\end{align}
A hyperelastic rod with stored energy density $\Phi$ is \emph{transversely hemitropic} if for all $\theta \in [0,2\pi)$, $\mfu, \mfv \in \bbR^3$,
$
	 \Phi(\msR_\psi \mfu, \msR_\psi \mfv) = \Phi(\mfu, \mfv). 
$
The rod is \emph{flip-symmetric} if it is transversely hemitropic and for all $\mfu, \mfv \in \bbR^3$,
$
\Phi (\msE \mfu, \msE \mfv) = \Phi(\mfu, \mfv),
$
and the rod is \emph{isotropic} if it is transversely hemitropic and for all 
$\mfu, \mfv \in \bbR^3$, 
$
	\Phi (\msE \mfu, -\msE \mfv) = \Phi(\mfu, \mfv).
$
Thus, our model with strain energy \eqref{eq:stenergy} is hemitropic and flip-symmetric. It is isotropic if and only if the twist-stretch coupling constant $\iota = 0$. We refer the reader to \cite{Healey02} for more on the notion of transverse symmetry for special Cosserat rods. 

\subsection{Monotonicity}

By using $L$ and $\gamma$ as our length and force scales, respectively, and appropriately nondimensionalizing the variables allows us to set 
\begin{align}
	L = 1, \quad \gamma = 1,
\end{align}
for the remainder of this paper. 

We now prove a mathematically attractive monotonicity property of our model.  
The following proposition implies that an increase in the bending couple $m_\mu$ accompanies an increase in the flexure $u_\mu$, an increase in twisting couple $m_3$ accompanies an increase in the twist $u_3$, an increase in shear force $n_\mu$ accompanies an increase in the shear strain $v_\mu$, and an increase in tension $n_3$ accompanies an increase in the dilatation strain $v_3$. 

\begin{prop}
	If $\mfu, \mfv \in \bbR^3$ and $Q(\mfu, \mfv) < 1$, then the Hessian of the stored energy density,
	\begin{align}
		D^2 W(\mfu, \mfv) = 
		\begin{bmatrix}
			\frac{\p \mfm}{\p \mfu}(\mfu,\mfv) & 	\frac{\p \mfm}{\p \mfv}(\mfu,\mfv) \\
			\frac{\p \mfn}{\p \mfu}(\mfu,\mfv) &
			\frac{\p \mfn}{\p \mfv}(\mfu, \mfv)
		\end{bmatrix},
	\end{align}
	is positive definite.  
\end{prop}

\begin{proof}
To establish the proposition, we compute  
\begin{align}
\frac{\p m_\mu}{\p u_\nu} &= \Bigl (1 - Q(\mfu, \mfv)^{p/2} \Bigr )^{-1/p-1} \\
&\quad\times \Bigl [
\bigl ( 1 - Q(\mfu, \mfv)^{p/2} \bigr )\al^2 \delta_{\mu \nu} + 
Q(\mfu, \mfv)^{p/2-1} (\al^2 u_\mu) (\al^2 u_\nu)
\Bigr  ], \\ 
\frac{\p m_\mu}{\p u_3} &= \frac{\p m_3}{\p u_\mu} \\
&=  \Bigl (1 - Q(\mfu, \mfv)^{p/2} \Bigr )^{-1/p-1} Q(\mfu,\mfv)^{p/2-1}  \al^2 u_\mu (\beta^2 u_3 + \iota (v_3-1)), 
\end{align}
as well as
\begin{align}
\frac{\p m_\mu}{\p v_\nu} &= \frac{\p n_\nu}{\p u_\mu} \\
&= \Bigl (1 - Q(\mfu, \mfv)^{p/2} \Bigr )^{-1/p-1} Q(\mfu, \mfv)^{p/2-1}   (\alpha^2 u_\mu) (\zeta^2 v_\nu), \\ 
\frac{\p m_\mu}{\p v_3} &= \frac{\p n_3}{\p u_\mu} \\
&= \Bigl (1 - Q(\mfu, \mfv)^{p/2} \Bigr )^{-1/p-1} Q(\mfu, \mfv)^{p/2-1} \alpha^2 u_\mu(\iota u_3 + \eta^2(v_3-1)),
\end{align}
and
\begin{align}
\frac{\p m_3}{\p u_3}  &= \Bigl (1 - Q(\mfu, \mfv)^{p/2} \Bigr )^{-1/p-1}  \\
&\quad \times \Bigl [
\bigl ( 1 - Q(\mfu, \mfv)^{p/2} \bigr )\beta^2  
Q(\mfu, \mfv)^{p/2-1} (\beta^2 u_3 + \iota(v_3-1))^2
\Bigr  ],  
\end{align}
and 
\begin{align}
\frac{\p m_3}{\p v_\nu} &= \frac{\p n_\nu}{\p u_3} \\
&= \Bigl (1 - Q(\mfu, \mfv)^{p/2} \Bigr )^{-1/p-1}   Q(\mfu, \mfv)^{p/2-1} (\beta^2 u_3 + \iota(v_3-1)) \zeta^2 v_\nu,  
\\ 
 \frac{\p m_3}{\p v_3} &= \frac{\p n_3}{\p u_3} \\
 &= \Bigl (1 - Q(\mfu, \mfv)^{p/2} \Bigr )^{-1/p-1} \\
&\quad \times \Bigl [
\bigl ( 1 - Q(\mfu, \mfv)^{p/2} \bigr )\iota \\ 
&\qquad + 
Q(\mfu, v_3-1)^{p/2-1} (\beta^2 u_3 + \iota(v_3-1))(\iota u_3 + \eta^2 (v_3-1))
\Bigr  ], 
\end{align}
and finally, 
\begin{align}
\frac{\p n_\mu}{\p v_\nu} &= \Bigl (1 - Q(\mfu, \mfv)^{p/2} \Bigr )^{-1/p-1} \\
&\quad\times \Bigl [
\bigl ( 1 - Q(\mfu, \mfv)^{p/2} \bigr )\zeta^2 \delta_{\mu \nu} + 
Q(\mfu, \mfv)^{p/2-1} (\zeta^2 v_\mu) (\zeta^2 v_\nu)
\Bigr  ], 
\\
\frac{\p n_\mu}{\p v_3} &= \frac{\p n_3}{\p v_\mu} \\
&=  \Bigl (1 - Q(\mfu, \mfv)^{p/2} \Bigr )^{-1/p-1} Q(\mfu,\mfv)^{p/2-1}  \zeta^2 v_\mu (\iota u_3 + \eta^2 (v_3-1)), 
\\ 
\frac{\p n_3}{\p v_3} &= 
\Bigl (1 - Q(\mfu, \mfv)^{p/2} \Bigr )^{-1/p-1} \\
&\times \Bigl [
\bigl ( 1 - Q(\mfu, \mfv)^{p/2} \bigr )\eta^2 + 
Q(\mfu, \mfv)^{p/2-1} (\iota u_3 + \eta^2 (v_3-1))^2
\Bigr  ]. 
\end{align} 
Then for all $\msa, \msb \in \bbR^3$, we have 
\begin{align}
\Bigl (1 &- Q(\mfu, \mfv) \Bigr )^{1/p+1}	
\begin{bmatrix}
	\msa \\
	\msb
\end{bmatrix}
  \cdot  D^2 W(\mfu, \mfv) \begin{bmatrix}
		\msa \\
			\msb
	\end{bmatrix} \\
&=
\bigl (1 - Q(\mfu, \mfv)^{p/2}) Q(\msa, \msb)
\\
&\quad + Q(\mfu, \mfv)^{p/2-1} \Bigl (
\al^2 a_\mu u_\mu + \zeta^2 b_\mu v_\mu +  a_3(\beta^2 u_3 + \iota (v_3-1)) 
\\&\hspace{1.5in}+ b_3(\iota u_3 + \eta^2 (v_3-1)) \Bigr )^2.
\end{align}
Since $Q(\cdot,\cdot)$ is positive definite, the proof is concluded. 
\end{proof}

\section{Some explicit equilibrium states}

In this section we discuss some solutions to \eqref{eq:neq} and \eqref{eq:meq} with no external body force or body couple and with an end thrust $\bs n(1)$ prescribed.
We assume that the rod is oriented so that $\bs n(1) = N \bs g_3$, where $N \in \bbR$. Then by \eqref{eq:neq}, for all $s \in [0,1]$
\begin{align}
	\bs n(s) = N \bs g_3. 
\end{align}   

Our discussion is from the semi-inverse standpoint rather than considering a fixed boundary value problem.   Some aspects of the equilibrium state will be fixed at the start in order to solve for the kinematic variables and obtain \emph{simple, explicit} solutions. In particular, we consider equilibrium states which are subject to either isolated contact forces or isolated contact couples, and in the latter case we consider only helical states.\footnote{In this work, a helical state is an equilibrium state such that the director $\bs d_3(s)$ has a constant polar angle in a fixed system of spherical coordinates.  The general study of helical states was initiated by Kirchhoff \cite{Kirchoff1859} for uniform, inextensible, unshearable rods with linear relations between the strains and components of the contact couple. Antman \cite{Antman74} and Chousiab-Maddocks \cite{Maddocks04} established some qualitative properties of helical states for fairly general constitutive relations describing a uniform special Cosserat rod.} 

\subsection{Euler angles and reduced equations}

As in previous works (see for example \cite{Love44}, \cite{WhitmanDeSilve74}, \cite{Antman74}) we compute explicit solutions by expressing the directors and equilibrium equations using Euler angles, defined as follows.   

Using spherical coordinates $(\varphi,\theta)$ for the sphere $\bbS^2$, we can write 
\begin{align}
	\bs d_3 = \sin \theta (\cos \varphi \bs g_1 + \sin \varphi \bs g_2) + \cos \theta \bs g_3.  \label{eq:d3eq}
\end{align}
We then define a new right-handed orthonormal basis $\{ \bs e_k \}$ via 
\begin{gather}
	\bs e_3 = \bs d_3, \bs e_2 = -\sin \varphi \bs g_1 + \cos \varphi \bs g_2, \\
	 \bs e_1 = \bs e_2 \times \bs e_3
	 = \cos \theta (\cos \varphi \bs g_1 + \sin \varphi \bs g_2 ) - \sin \theta \bs g_3.  
\end{gather}
Since $\bs d_1, \bs d_2$ are orthogonal to $\bs d_3 = \bs e_3$, there exists $\psi \in \bbR$ such that 
\begin{align}
	\bs d_1 &= \cos \psi \bs e_1 + \sin \psi \bs e_2, \label{eq:d1eq}\\
	\bs d_2 &= -\sin \psi \bs e_1 + \cos \psi \bs e_2. \label{eq:d2eq}
\end{align}
The angles $(\varphi,\theta, \psi)$ are referred to as the Euler angles parameterizing the directors $\{ \bs d_k \}$. We can then express $\bs n = N \bs g_3$ as 
\begin{align}
	\bs n &= -N \sin \theta \cos \psi \bs d_1 + N \sin \theta \sin \psi \bs d_2 + N \cos \theta \bs d_3. \label{eq:nsoln}
\end{align}
In the basis $\{ \bs e_k \}$, the contact couple $\bs m = m_k \bs d_k = M_k \bs e_k$ with
\begin{align}
	M_1 &= m_1 \cos \psi - m_2 \sin \psi, \label{eq:M1eq} \\
	M_2 &= m_1 \sin \psi + m_2 \cos \psi, \label{eq:M2eq} \\
	M_3 &= m_3, \label{eq:M3eq}
\end{align}
and the contact force $\bs n = N_k \bs e_k$ with 
\begin{align}
	N_1 &= -N \sin \theta, \label{eq:N1eq} \\
	N_2 &= 0, \label{eq:N2eq} \\
	N_3 &= N \cos \theta. \label{eq:N3eq} 
\end{align}

Using the expressions for the directors in terms of the Euler angles, one readily verifies the following relationships between the strains and Euler angles: 
\begin{align}
	u_1 &= \theta' \sin \psi - \varphi' \sin \theta \cos \psi, \label{eq:u1angle} \\ 
	u_2 &= \theta_s \cos \psi + \varphi' \sin \theta \sin \psi, \label{eq:u2angle} \\
	u_3 &= \psi' + \varphi' \cos \theta. \label{eq:u3angle}
\end{align}
If we denote 
\begin{align}
	u &= \Bigl ( 1 +  Q^*(\mfm, \mfn)^{p/2} \Bigr )^{-1/p} \frac{1}{\al^2}, \\ 
	v &= \Bigl ( 1 +  Q^*(\mfm, \mfn)^{p/2} \Bigr )^{-1/p} \frac{1}{\zeta^2}, 
\end{align}
then $u_\mu = u \cdot m_\mu$, $v_\mu = v \cdot n_\mu$, $\mu = 1,2$. Then \eqref{eq:u1angle}, \eqref{eq:u2angle}, \eqref{eq:u3angle}, \eqref{eq:M1eq}, and \eqref{eq:M2eq} imply that
\begin{align}
	\sin \theta \varphi' &= - u M_1, \label{eq:phiseq} \\
	\theta' &= - u M_2, \label{eq:thetaseq} \\
	\psi' + \cos \theta \varphi' &= u_3. \label{eq:psiseq} 
\end{align}
Expressing \eqref{eq:meq} in the basis $\{ \bs e_k \}$ using \eqref{eq:N1eq}, \eqref{eq:N2eq}, \eqref{eq:N3eq} yields
\begin{align}
	M_1' - M_2 \cos \theta \varphi' + \theta' M_3 &= 0, \label{eq:M1seq} \\
	M_2' + (M_1 \cos \theta + M_3 \sin \theta) \varphi' &= N v_3 \sin \theta - N^2 v \cos \theta \sin \theta, \label{eq:M2seq} \\
	M_3' &= 0. \label{eq:M3seq}
\end{align}
The previous six equations are for the six unknowns $\varphi, \theta, \psi, M_1, M_2, M_3$.

\subsection{Rod subject to isolated contact forces} 

In this subsection we will consider the case when only contact forces are present: for all $s \in [0,1]$
\begin{align}
\bs m(s) = \bs 0 \iff M_1(s) = M_2(s) = M_3(s) = 0. 
\end{align} 
By the equilibrium equations, an equilibrium state under pure contact forces exists if and only if 
\begin{align}
	\sin \theta \varphi' &= 0, \label{eq:TTvarphi}\\
	\theta' &= 0, \\
	\psi' + \cos \theta \varphi' &= u_3 \\
	N \sin \theta(v_3 - v N \cos \theta) &= 0. \label{eq:TTtheta} 
\end{align}
Thus, $\theta$ is a constant and is determined by $N$ via \eqref{eq:TTtheta}.

\begin{prop}
The center line of an equilibrium state satisfying \eqref{eq:TTtheta} is parallel to $\bs g_3$.
\end{prop} 

\begin{proof}
We compute the following inner products using $\bs r' = v_k \bs d_k$ along with $v_\mu = v \cdot n_\mu$ and \eqref{eq:nsoln}:  
\begin{align}
	\bs r' \cdot \bs g_1 &= \sin \theta \cos \varphi (v_3 - v N \cos \theta)
	= 0, \\
	\bs r' \cdot \bs g_2 &= \sin \theta \sin \varphi (v_3 - v N \cos \theta )
	= 0,
\end{align}
Thus, $\bs r'$ is parallel to $\bs g_3$. 
\end{proof}

We consider two cases for \eqref{eq:TTtheta}: 
\begin{align}
	\sin \theta = 0 \quad \mbox{or} \quad v_3 = vN \cos \theta. \label{eq:TTalt}
\end{align}
For the first case, we consider only the sub-case $\theta = 0$ (the sub-case $\theta = \pi$ can be analyzed similarly) from which it follows that $\bs d_3 = \bs g_3$ and $\bs n = N \bs d_3 = N \bs g_3$. We immediately conclude the following.  

\begin{prop}
Assume that $\theta = 0$ and, after a proper rotation of the plane spanned by $\{ \bs g_1, \bs g_2\}$ if necessary, that $\varphi(0) = 0$. Then the strains of the associated equilibrium state are constant and given by  
\begin{align}
	u_1 &= u_2 = 0, \\
	u_3 &= \Bigl ( 1 + \frac{\beta^p}{(\beta^2 \eta^2 - \iota^2)^{p/2}} |N|^p \Bigr )^{-1/p} \frac{-\iota N}{\beta^2 \eta^2 - \iota^2},
	\label{eq:TTu3}
	\\
	v_1 &= v_2 = 0, \\
		v_3 - 1 &= \Bigl ( 1 + \frac{\beta^p}{(\beta^2 \eta^2 - \iota^2)^{p/2}} |N|^p \Bigr )^{-1/p} \frac{\beta^2 N}{\beta^2 \eta^2 - \iota^2}.
\end{align}
The center line of the rod is parallel to $\bs g_3$,
\begin{align}
	\bs r'(s) = (v_1^2 + v_2^2 + v_3^2)^{1/2} \bs g_3,
\end{align}
 and the directors are given by 
	\begin{align}
		\bs d_1(s) &= \cos (u_3 s + \psi(0)) \bs g_1 + \sin (u_3 s + \psi(0)) \bs g_2, \\ 	
		\bs d_2(s) &= 
		-\sin (u_3 s + \psi(0)) \bs g_1 + \cos (u_3 s + \psi(0)) \bs g_2, \\
		\bs d_3(s) &= \bs g_3.
	\end{align}
\end{prop}

We emphasize that if $\theta = 0$, then the configuration is unsheared, $v_1 = v_2 = 0,$ and if $\iota \neq 0$, then the rod is twisted, $u_3 \neq 0$, under an isolated end thrust.  In accordance with common experiences with ropes, threads, or dishtowels unwinding when stretched, and because of \eqref{eq:TTu3}, we interpret $\iota > 0$ as modeling right-handed chirality and $\iota < 0$ as modeling left-handed chirality. As $N \rar \pm \infty$, we obtain the nonzero limiting strains: 
\begin{align}
	\lim_{N \rar \pm \infty} u_3 = \mp \frac{\iota}{\beta}(\beta^2 \eta^2 - \iota^2)^{-1/2}, \\
	\lim_{N \rar \pm \infty} (v_3 - 1) = \pm \beta (\beta^2 \eta^2 - \iota^2)^{-1/2}. \label{eq:maxv3}
\end{align} 

We now consider the case 
\begin{align}
	v_3 = v N \cos \theta, \quad \theta \in (0,\pi). \label{eq:TTcase2}
\end{align}
Since $v$ and $v_3$ are positive we must have $n_3 = N\cos \theta > 0$. Then the rod must be in a tensile state, $v_3 - 1 > 0$. We will focus on the sub-case $\theta \in (0,\pi/2)$ (the other sub-case $\theta \in (\pi/2,\pi)$ can be analyzed similarly). This is equivalent to assuming that $N > 0$. 

We have that \eqref{eq:TTcase2} is equivalent to 
\begin{align}
	v_3 &= \Bigl [ 1 + \Bigl (
	\frac{1}{\zeta^2} N^2 \sin^2 \theta + \frac{\beta^2}{\beta^2\eta^2 - \iota^2} N^2 \cos^2 \theta 
	\Bigr )^{p/2} \Bigr ]^{-1/p} \frac{n_3}{\zeta^2} \\
	&= \frac{\beta^2 \eta^2 - \iota^2}{\beta^2 \zeta^2} (v_3 - 1),
\end{align}
which is equivalent to 
\begin{align}
	v_3 - 1 = \Bigl ( \frac{\beta^2 \eta^2 - \iota^2}{\beta^2 \zeta^2} - 1 \Bigr )^{-1} \label{eq:TTv3}. 
\end{align}
Thus, a necessary condition for an equilibrium state satisfying \eqref{eq:TTcase2} to exist is that the material moduli satisfy 
\begin{gather}
	 \frac{\beta^2 \eta^2 - \iota^2}{\beta^2 \zeta^2} - 1 > 0
	\iff \frac{1}{\zeta^2} - \frac{\beta^2}{\beta^2 \eta^2 - \iota^2} > 0 \iff
	\eta^2 > \zeta^2 + \frac{\iota^2}{\beta^2}. \label{eq:TTin1}
\end{gather}
By \eqref{eq:maxv3} we also expect that a necessary condition for \eqref{eq:TTv3} being satisfied is that the limiting value for $v_3 - 1$ is strictly bigger than the right hand side of \eqref{eq:TTv3}:
\begin{gather}
	\Bigl ( \frac{\beta^2 \eta^2 - \iota^2}{\beta^2 \zeta^2} - 1 \Bigr )^{-1} < \frac{\beta}{(\beta^2 \eta^2 - \iota^2)^{1/2}} \label{eq:TTin2} \\
	\iff 
\Bigl (
\frac{\beta^2 \eta^2 - \iota^2}{\beta^2 \zeta^2} - 1
\Bigr )^p \frac{\beta^{2p}}{(\beta^2 \eta^2 - \iota^2)^p} > 
\frac{\beta^p}{(\beta^2 \eta^2 - \iota^2)^{p/2}}.
\end{gather}

We now prove that \eqref{eq:TTin1} and \eqref{eq:TTin2} are sufficient to solve \eqref{eq:TTv3} uniquely for  $\theta = \hat \theta(N) \in (0,\pi/2)$, for all $N > N_{thresh}$ with $N_{thresh}$ defined in Proposition \ref{p:bif}. In particular, it immediately follows that 
\begin{align}
	\{ (N, \hat \theta(N)) \mid N > N_{thresh} \}
\end{align} 
is a nontrivial branch of the set of solutions to \eqref{eq:TTtheta} bifurcating from the trivial branch $\{ (N, 0) \mid N \geq 0 \}$ of tensile solutions to \eqref{eq:TTtheta} at the bifurcation point $(N_{thresh}, 0)$. This nontrivial branch can be interpreted as a \emph{shearing instability} of the rod for large tensile end thrusts. 

\begin{prop}\label{p:bif}
Assume that the material moduli satisfy  \eqref{eq:TTin1} and \eqref{eq:TTin2}, and let $N_{thresh} > 0$ be defined via 
\begin{align}
	N_{thresh}^{-p} = 
	\Bigl (
	\frac{\beta^2 \eta^2 - \iota^2}{\beta^2 \zeta^2} - 1
	\Bigr )^p \frac{\beta^{2p}}{(\beta^2 \eta^2 - \iota^2)^p} - 
	\frac{\beta^p}{(\beta^2 \eta^2 - \iota^2)^{p/2}}. \label{eq:N0def}
\end{align}
If $N > N_{thresh}$ then
there exists a unique $\hat \theta(N) \in (0,\pi/2)$ such that
\begin{align}
	v_3 - 1 &= \Bigl [ 1 + \Bigl (
	\frac{1}{\zeta^2} N^2 \sin^2 \theta + \frac{\beta^2}{\beta^2\eta^2 - \iota^2} N^2 \cos^2 \theta 
	\Bigr )^{p/2} \Bigr ]^{-1/p} \\
	&\quad \times \frac{\beta^2}{\beta^2 \eta^2 - \iota^2} N \cos \theta \label{eq:TTv3def}
\end{align} 
with $\theta = \hat \theta(N)$ satisfies \eqref{eq:TTv3}. Moreover, 
\begin{align}
 \lim_{N \rar \infty} \hat \theta(N) = 
	\cos^{-1} \Bigl \{
	\Bigl [
	\Bigl (
	\frac{1}{\zeta^2} - \frac{\beta^2}{\beta^2 \eta^2 - \iota^2} 
	\Bigr )^2 + 
	\frac{1}{\zeta^2} - \frac{\beta^2}{\beta^2 \eta^2 - \iota^2} 
	\Bigr ]^{-1/2} \frac{1}{\zeta} \Bigr \}. \label{eq:thetainfinity}
\end{align}
\end{prop}

\begin{proof}
If $p = 2$, one can solve for $\hat \theta(N) \in (0, \pi/2)$ explicitly:
\begin{align}
\hat \theta(N) = \cos^{-1} \Bigl \{
\Bigl [
\Bigl (
\frac{1}{\zeta^2} - \frac{\beta^2}{\beta^2 \eta^2 - \iota^2} 
\Bigr )^2 + 
\frac{1}{\zeta^2} - \frac{\beta^2}{\beta^2 \eta^2 - \iota^2} 
\Bigr ]^{-1/2} \Bigl (
\frac{1}{N^2} + \frac{1}{\zeta^2}
\Bigr )^{1/2} \Bigr \}
\end{align}
In general, we write \eqref{eq:TTv3} as 
\begin{align}
	f_N(\cos \theta) = \Bigl (
	\frac{\beta^2 \eta^2 - \iota^2}{\beta^2 \zeta^2} - 1
	\Bigr )^{-1}
\end{align}
where 
\begin{align}
f_N(x) &= \Bigl [
1 + \Bigl ( \frac{N^2}{\zeta^2}(1-x^2) + \frac{\beta^2}{\beta^2 \eta^2 - \iota^2} N^2 x^2 
\Bigr )^{p/2}
\Bigr ]^{-1/p} \frac{N \beta^2}{\beta^2 \eta^2 - \iota^2} x \\
&= \Bigl [
N^{-p} + \Bigl ( \frac{1}{\zeta^2}(1-x^2) + \frac{\beta^2}{\beta^2 \eta^2 - \iota^2} x^2 
\Bigr )^{p/2}
\Bigr ]^{-1/p} \frac{\beta^2}{\beta^2 \eta^2 - \iota^2} x, 
\end{align}
for $x \in [0,1]$. By \eqref{eq:TTin1} we conclude that 
\begin{align}
	f'_N(x) &= \Bigl [
	N^{-p} + \Bigl ( \frac{1}{\zeta^2}(1-x^2) + \frac{\beta^2}{\beta^2 \eta^2 - \iota^2} x^2 
	\Bigr )^{p/2}
	\Bigr ]^{-1/p-1} \frac{\beta^2}{\beta^2 \eta^2 - \iota^2}\\
	&\quad \times \Bigl [
	1 + \Bigl ( \frac{1}{\zeta^2}(1-x^2) + \frac{\beta^2}{\beta^2 \eta^2 - \iota^2} x^2 
	\Bigr )^{p/2-1}\Bigl (
	\frac{1}{\zeta^2} - \frac{\beta^2}{\beta^2 \eta^2 - \iota^2}
	\Bigr ) x^2 
	\Bigr ] 
	 > 0.
\end{align}
Thus, $f_N : [0,1] \rar [0, f_N(1)]$ is an increasing bijection. We have 
\begin{align}
	f_N(1) = 
	\Bigl (
	N^{-p} + \frac{\beta^p}{(\beta^2 \eta^2 - \iota^2)^{p/2}}
	\Bigr )^{-1/p} \frac{\beta^2}{\beta^2 \eta^2 - \iota^2} \rar 
	\beta (\beta^2 \eta^2 - \iota^2)^{-1/2}
\end{align}
as $N \rar \infty$. We conclude that for all $N > N_{thresh}$ where $N_{thresh}$ satisfies \eqref{eq:N0def}, we have 
\begin{align}
0 < \Bigl (
\frac{\beta^2 \eta^2 - \iota^2}{\beta^2 \zeta^2} - 1
\Bigr )^{-1} <
	\Bigl (
N^{-p} + \frac{\beta^p}{(\beta^2 \eta^2 - \iota^2)^{p/2}}
\Bigr )^{-1/p} \frac{\beta^2}{\beta^2 \eta^2 - \iota^2} = f_N(1).
\end{align}
Thus, there exists a unique $x_N \in (0,1)$ such that 
\begin{align}
	f_N(x_N) = \Bigl (
	\frac{\beta^2 \eta^2 - \iota^2}{\beta^2 \zeta^2} - 1
	\Bigr )^{-1}
\end{align}
Then $\hat \theta(N) = \cos^{-1} x_N \in (0, \pi/2)$ is the desired angle. 

We define 
\begin{align}
	\theta(\infty) = 	\cos^{-1} \Bigl \{
	\Bigl [
	\Bigl (
	\frac{1}{\zeta^2} - \frac{\beta^2}{\beta^2 \eta^2 - \iota^2} 
	\Bigr )^2 + 
	\frac{1}{\zeta^2} - \frac{\beta^2}{\beta^2 \eta^2 - \iota^2} 
	\Bigr ]^{-1/2} \frac{1}{\zeta} \Bigr \}.
\end{align}
Let $\{ N_k \}$ be a sequence with $N_k \rar \infty$ as $k \rar \infty$, and let 
\begin{align}
	x_{N_k} = \cos \theta(N_k).
\end{align}
 Since $x_{N_k} \in [0,1]$ for all $k$, there exists a subsequence of $\{N_k\}$ denoted by $\{N_j\}$ and $x \in [0,1]$ such that $x_{N_j} \rar x$ as $j \rar \infty$.  Then  
 \begin{align}
 		\Bigl ( \frac{1}{\zeta^2}(1-x^2) + \frac{\beta^2}{\beta^2 \eta^2 - \iota^2} x^2 
 	\Bigr )^{-1/2} \frac{\beta}{\beta^2 \eta^2 - \iota^2} x &= 
 	\lim_{j \rar \infty}
 	f_{N_j}(x_{N_j}) \\ 
 	&= \Bigl (
 	\frac{\beta^2 \eta^2 - \iota^2}{\beta^2 \zeta^2} - 1
 	\Bigr )^{-1}.
 \end{align}
Solving for $x$, we conclude that $x = \cos \hat \theta(\infty)$ and $\cos \theta(N_j) = x_{N_j} \rar \cos \hat \theta(\infty)$. Thus,  
\begin{align}
	\lim_{j \rar \infty} \theta(N_j) = \hat \theta(\infty). 
\end{align}
Since every sequence $\{ N_k \}$ with $N_k \rar \infty$ has a subsequence $\{N_j\}$ with $\theta(N_j) \rar \hat \theta(\infty)$, we conclude \eqref{eq:thetainfinity}. 
\end{proof}

For this branch of nontrivial solutions, $\sin \theta \neq 0$ and \eqref{eq:TTvarphi} imply that $\varphi' = 0$.  Moreover, via \eqref{eq:TTv3def} and \eqref{eq:TTv3} we have the identity 
\begin{gather}
	\Bigl [ 1 + \Bigl (
	\frac{1}{\zeta^2} N^2 \sin^2 \theta + \frac{\beta^2}{\beta^2\eta^2 - \iota^2} N^2 \cos^2 \theta 
	\Bigr )^{p/2} \Bigr ]^{-1/p} \\
	= \frac{\beta^2 \eta^2 - \iota^2}{\beta^2} \frac{1}{N \cos \theta} 
	\Bigl (
	\frac{\beta^2 \eta^2 - \iota^2}{\beta^2 \zeta^2} - 1
	\Bigr )^{-1}. \label{eq:bifidentity}
\end{gather}
Using \eqref{eq:bifidentity} we now summarize the properties of the branch of sheared tensile states. See Figure \ref{f:fig3} for a description of the two tensile branches' configurations compared to the rod's unstressed configuration. 

\begin{figure}[t]
	\centering
	\includegraphics[width=\textwidth]{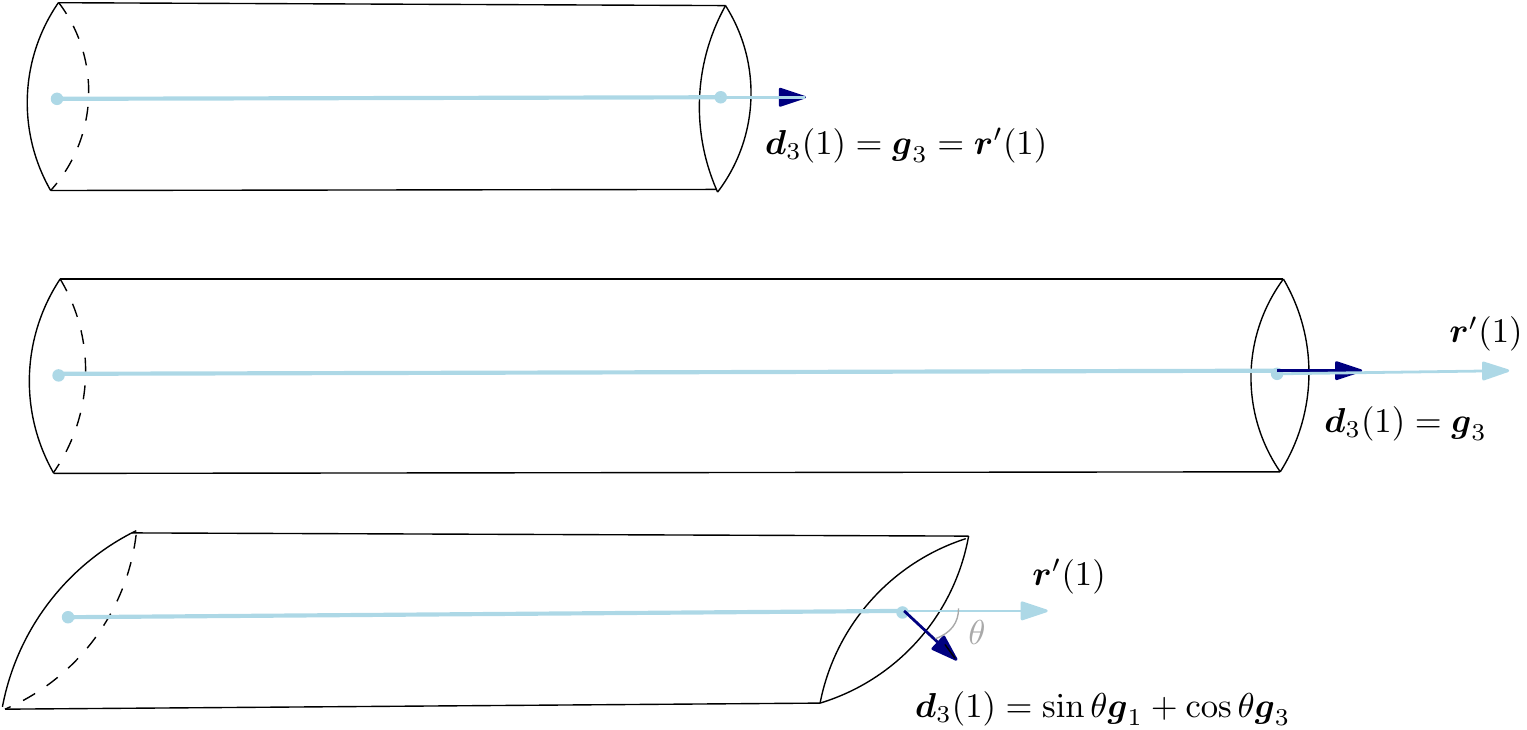}
	\caption{The Carolina blue vector represents $\bs r'(1)$ and the navy vector represents $\bs d_3(1)$ for each configuration. The top configuration represents the unstressed state of the rod. The second and third configurations from the top qualitatively describe the two branches of equilibrium states subject to a large enough isolated, tensile end thrust parallel to $\bs r'(1)$.}
	\label{f:fig3}
\end{figure}

\begin{prop}
Assume, after a proper rotation of the plane spanned by $\{ \bs g_1, \bs g_2 \}$ if necessary, that $\varphi(0) = 0$ (so then $\varphi(s) = 0$ for all $s$). 
Let $N > N_{thresh}$ and $\theta = \hat \theta(N) \in (0,\pi/2)$ be as in Proposition \ref{p:bif}. Then the strains of the associated equilibrium state are given by 
\begin{align}
	u_1 &= u_2 = 0, \\
	u_3 &= \Bigl (
	\frac{\beta^2 \eta^2 - \iota^2}{\beta^2 \zeta^2} - 1
	\Bigr )^{-1} \frac{-\iota}{\beta^2}, \\
	v_1 &= -\Bigl (
	\frac{\beta^2 \eta^2 - \iota^2}{\beta^2 \zeta^2} - 1
	\Bigr )^{-1} \frac{\beta^2 \eta^2 - \iota^2}{\beta^2 \zeta^2} \tan \theta \cos (u_3s + \psi(0)), \\
	v_2 &= \Bigl (
	\frac{\beta^2 \eta^2 - \iota^2}{\beta^2 \zeta^2} - 1
	\Bigr )^{-1} \frac{\beta^2 \eta^2 - \iota^2}{\beta^2 \zeta^2} \tan \theta \sin (u_3 s + \psi(0)), \\
	v_3 - 1 &= \Bigl (
	\frac{\beta^2 \eta^2 - \iota^2}{\beta^2 \zeta^2} - 1
	\Bigr )^{-1}.
\end{align}
The center line of the rod is parallel to $\bs g_3$, and the directors are given by 
	\begin{align}
	\bs d_1(s) &= \cos \theta \cos(u_3s + \psi(0)) \bs g_1 + 
	\sin(u_3 s + \psi(0)) \bs g_2 \\ &\quad - \sin \theta \cos (u_3 s + \psi(0)) \bs g_3, \\ 	
	\bs d_2(s) &= -\cos \theta \sin(u_3s + \psi(0)) \bs g_1 + 
	\cos(u_3 s + \psi(0)) \bs g_2 \\ &\quad + \sin \theta \sin (u_3 s + \psi(0)) \bs g_3, \\
	\bs d_3(s) &= \sin \theta \bs g_1 + \cos \theta \bs g_3.
\end{align}
\end{prop}

\subsection{Rod subject to isolated contact couples}
We conclude this section by considering certain equilibrium states under isolated contact couples, $N = 0$. In particular, we will consider a special class of such equilibrium states, those satisfying: for all $s \in [0,1]$
\begin{align}
M_2 = 0 \iff \theta' = 0. 
\end{align}
As in the previous subsection,  for simplicity we will assume that
\begin{align}
	\theta \in [0,\pi/2].
\end{align}
The case $\theta \in (\pi/2,\pi]$ can be analyzed similarly. 
If $\theta' = M_2 = 0$, then by \eqref{eq:M1seq} and \eqref{eq:M3seq}, we conclude that $\varphi$ and $\psi$ satisfy
\begin{align}
	\sin \theta \varphi' &= -u M_1, \label{eq:noN1} \\
	\psi' + \cos \theta \varphi' &= u_3, \\
	(M_1 \cos \theta + M_3 \sin \theta) \varphi' &= 0,
\end{align}
where $\theta, M_1, M_3$ are constant. 

If $M_2 = 0$ and $M_1 = 0$ then by \eqref{eq:M1eq} and \eqref{eq:M2eq} we conclude that $m_\mu = 0$ for $\mu = 1,2$. By \eqref{eq:noN1} we conclude that either: $\theta = 0$, or $\theta \neq 0$ and $\varphi' = 0$. We then have the following.

\begin{prop}
Assume that $M_2 = 0$, $M_1 = 0$ and, after a proper rotation of the plane spanned by $\{\bs g_1, \bs g_2\}$ if necessary, that $\varphi(0) = 0$. 
Then the strains of the associated equilibrium state are constant and given by
\begin{align}
	u_1 &= u_2 = 0, \\
	u_3 &= \Bigl ( 1 + \frac{\eta^p}{(\beta^2 \eta^2 - \iota^2)^{p/2}} |M_3|^p \Bigr )^{-1/p} \frac{\eta^2 M_3}{\beta^2 \eta^2 - \iota^2},
	\label{eq:TTu3M}
	\\
	v_1 &= v_2 = 0, \\
	v_3 - 1 &= \Bigl ( 1 + \frac{\eta^p}{(\beta^2 \eta^2 - \iota^2)^{p/2}} |M_3|^p \Bigr )^{-1/p} \frac{-\iota M_3}{\beta^2 \eta^2 - \iota^2}.
\end{align}
The director $\bs d_3$ is constant, the center line of the rod is parallel to $\bs d_3$, 
and the directors are given by 
		\begin{align}
		\bs d_1(s) &= \cos \theta \cos(u_3s + \psi(0)) \bs g_1 + 
		\sin(u_3 s + \psi(0)) \bs g_2 \\ &\quad - \sin \theta \cos (u_3 s + \psi(0)) \bs g_3, \\ 	
		\bs d_2(s) &= -\cos \theta \sin(u_3s + \psi(0)) \bs g_1 + 
		\cos(u_3 s + \psi(0)) \bs g_2 \\ &\quad + \sin \theta \sin (u_3 s + \psi(0)) \bs g_3, \\
		\bs d_3(s) &= \sin \theta \bs g_1 + \cos \theta \bs g_3.
	\end{align}
\end{prop}

We comment that the length of the rod under the isolated couple $\bs m(s) = M_3 \bs d_3$ is given by 
\begin{align}
	\int_0^1 |\bs r'(s)| ds - 1 = v_3 - 1 =  \Bigl ( 1 + \frac{\eta^p}{(\beta^2 \eta^2 - \iota^2)^{p/2}} |M_3|^p \Bigr )^{-1/p} \frac{-\iota M_3}{\beta^2 \eta^2 - \iota^2}. 
\end{align}
Thus, if $\iota \neq 0$, then the rod changes length due to a isolated couple. In particular, if the chirality of the rod is opposite to the chirality of the couple, $-\iota M_3 > 0$, then we observe a \emph{Poynting} effect: the application of an isolated couple elongates the rod.  As $M_3 \rar \pm \infty$ we obtain the nonzero limiting strains: 
\begin{align}
	\lim_{M_3 \rar \pm \infty} u_3 &= \pm \eta (\beta^2 \eta^2 - \iota^2)^{-1/2}, \\
	\lim_{M_3 \rar \pm \infty} (v_3 - 1) &= \mp \frac{\iota}{\eta} (\beta^2 \eta^2 - \iota^2)^{-1/2}. 
\end{align}

We conclude this study by considering the case $M_1 \neq 0$. We note that by \eqref{eq:phiseq}, there exists $s \in [0,1]$ such that $\sin \theta \varphi'(s) = 0$ if and only if $M_1 = 0$.  

\begin{prop}
	Assume that $M_2 = 0$, $M_1 \neq 0$,  $M_3 = -M_1 \cot \theta$ with $\theta \in (0,\pi/2]$, and, after a proper rotation of the plane spanned by $\{\bs g_1, \bs g_2\}$ if necessary, that $\varphi(0) = 0$. Define 
	\begin{align}
			\varphi(s) &= 
		- \Bigl [
		1 + |M_1|^p \Bigl (
		\frac{1}{\al^2} + \frac{\eta^2}{\beta^2 \eta^2 - \iota^2} \cot^2 \theta 
		\Bigr )^{p/2}
		\Bigr ]^{-1/p} M_1 (\csc \theta) s, \\
		\psi(s) &= -\Bigl (
		1 - \frac{\beta^2 \eta^2 - \iota^2}{\al^2 \eta^2}
		\Bigr ) \Bigl [
		1 + |M_1|^p \Bigl (
		\frac{1}{\al^2} + \frac{\eta^2}{\beta^2 \eta^2 - \iota^2} \cot^2 \theta 
		\Bigr )^{p/2}
		\Bigr ]^{-1/p}\\
		&\quad \times \frac{\eta^2}{\beta^2 \eta^2 - \iota^2} M_1 (\cot \theta) s + \psi(0),
\end{align}
	Then the strains of the associated equilibrium state are given by
\begin{align}
	u_1 &= \Bigl [
	1 + |M_1|^p \Bigl (
	\frac{1}{\al^2} + \frac{\eta^2}{\beta^2 \eta^2 - \iota^2} \cot^2 \theta 
	\Bigr )^{p/2}
	\Bigr ]^{-1/p} \frac{1}{\al^2} M_1 \cos \psi(s), \\
	u_2 &= -\Bigl [
	1 + |M_1|^p \Bigl (
	\frac{1}{\al^2} + \frac{\eta^2}{\beta^2 \eta^2 - \iota^2} \cot^2 \theta 
	\Bigr )^{p/2}
	\Bigr ]^{-1/p} \frac{1}{\al^2} M_1 \sin \psi(s), \\
	u_3 &= - \Bigl [
	1 + |M_1|^p \Bigl (
	\frac{1}{\al^2} + \frac{\eta^2}{\beta^2 \eta^2 - \iota^2} \cot^2 \theta 
	\Bigr )^{p/2}
	\Bigr ]^{-1/p} \frac{\eta^2}{\beta^2 \eta^2 - \iota^2} M_1 \cot \theta, \\
	v_1 &= v_2 = 0, \\
	v_3 - 1 &= 
	\Bigl [
	1 + |M_1|^p \Bigl (
	\frac{1}{\al^2} + \frac{\eta^2}{\beta^2 \eta^2 - \iota^2} \cot^2 \theta 
	\Bigr )^{p/2}
	\Bigr ]^{-1/p} \frac{\iota}{\beta^2 \eta^2 - \iota^2} M_1 \cot \theta,
\end{align} 
The center line is tangent to the director $\bs d_3$, 
\begin{align}
\bs r'(s) = v_3 \bs d_3(s) = v_3\sin \theta \cos \varphi(s) \bs g_1 + v_3 \sin \theta \cos \varphi(s) \bs g_2 + v_3 \cos \theta \bs g_3, \label{eq:rprimehelix}
\end{align}
\end{prop}

As discussed by Antman within a more general setting \cite{Antman74}, we note that \eqref{eq:rprimehelix} implies that the center line $\bs r(s)$  is a right-handed helix of radius $a = \frac{v_3 \sin \theta}{\varphi'}$ and pitch $b = v_3 \cos \theta$.  

In the case $\theta = \pi/2$, $M_3 = 0$, $\bs r$ is a circle in the plane 
spanned by $\{\bs g_1, \bs g_2\}$, 
\begin{align}
	\bs r(s) &= \Bigl ( 1 + |M_1|^p \al^{-p} \Bigr )^{1/p} \al^2 M_1^{-1} \sin \Bigl 
	(
	\Bigl ( 1 + |M_1|^p \al^{-p} \Bigr )^{-1/p} \al^{-2} M_1 s
	\Bigr  )
	 \bs g_1 \\
	 &\quad + \Bigl ( 1 + |M_1|^p \al^{-p} \Bigr )^{1/p} \al^2 M_1^{-1}
	 \cos 	\Bigl (
	 \Bigl ( 1 + |M_1|^p \al^{-p} \Bigr )^{-1/p} \al^{-2} M_1 s
	 \Bigr  ) \bs g_2,
\end{align}
and the rod is in a state of pure bending,
\begin{align}
	u_1 &= \Bigl ( 1 + |M_1|^p \al^{-p} \Bigr )^{-1/p} \frac{1}{\al^2} M_1 \cos \psi(0), \\
	u_2 &= -\Bigl ( 1 + |M_1|^p \al^{-p} \Bigr )^{-1/p} \frac{1}{\al^2} M_1 \sin \psi(0), \\
	u_3 &= 0, \\
	v_1 &= v_2 = v_3 - 1 = 0. 
\end{align}
As $M_1 \rar \pm \infty$ we obtain the limiting strains and curvature
\begin{align}
	\lim_{M_1 \rar \pm \infty} u_1 &= \pm \frac{1}{\al} \cos \psi(0), \\
	\lim_{M_1 \rar \pm \infty} u_2 &= \mp \frac{1}{\al} \sin \psi(0), \\
	\lim_{M_1 \rar \pm \infty} (u_1^2 + u_2^2)^{1/2} &= \frac{1}{\al}.  
\end{align}

\section{Conclusion}

In this work we introduced an intrinsic set of strain-limiting constitutive relations, between the geometrically exact strains and components of the contact couple and force, for special Cosserat rods. We then showed these relations are derivable from a complementary energy, are orientation preserving and satisfy a mathematically attractive monotonicity property. Finally, we computed some explicit equilibrium states displaying rich behavior including Poynting effects and tensile shearing bifurcations. We now discuss potential mathematical and modeling perspectives for future work.   

\subsection{Mathematical perspectives}

In Section 4 we computed equilibrium states from the semi-inverse standpoint and not by considering a fixed boundary value problem. 
Upon fixing boundary conditions at the two ends, the following problems are then suggested:
\begin{itemize}
	\item multiplicity of equilibrium states,
	\item stability of these equilibrium states. 
\end{itemize}   

Needless to say, both of these problems are highly dependent on the choice of boundary conditions. Even in the case of a straight center line, the equilibrium state corresponding to the shearing bifurcation obtained in Proposition \ref{p:bif} is not present if one imposes a fixed position and orientation at the end $s = 0$, 
\begin{gather}
	\bs r(0) = \bs o, \quad \bs d_k(0) = \bs g_k, \quad k = 1, 2, 3, \\
	\bs m(1) = \bs 0, \quad \bs n(1) = N \bs g_3. 
\end{gather}
However, if instead one imposes that the resultant contact couple about the origin $\bs o$ is zero,
\begin{gather}
	\bs r(0) = \bs o, \quad \bs m(1) + (\bs r(1) - \bs o) \times \bs n(1) - \bs m(0) = \bs 0, \quad \bs n(1) = N \bs g_3,	
\end{gather} 
then both branches of equilibrium states are present. 

In our opinion, stability of these equilibrium states is a dynamic question and cannot even be formulated unless the time-dependent field equations and nature of dissipation are specified.  For example, in the isothermal setting and interpreting the directors $\{ \bs d_1, \bs d_2 \}$ as specifying the material cross sections' principal axes of inertia, the dynamic field equations are given by 
\begin{align}
	(\rho A) \p_t^2 \bs r &= \p_s \bs n, \label{eq:linmom} \\
	\p_t \bigl [(\rho \bs J) \bs w \bigr ] &= \p_s \bs m + \p_s \bs r \times \bs n, \label{eq:angmom}
\end{align}
where $(\rho A)(s)$ is the mass density per unit reference length and $(\rho \bs J)(s,t)(\bs w(s,t))$ is the angular momentum of the of the cross section relative to $\bs r$ (calculable from given quantities). See Chapter 8 of \cite{AntmanBook}. 

However, there are infinitely many choices of dissipative mechanisms, each corresponding to a specification of a nonnegative total dissipation rate.  Two distinct natural choices which reduce to the constitutive relations introduced in this paper in the static setting are 
\begin{gather}
	\msu + \al \p_t \msu = \frac{\p W^*}{\p \msm}(\msm, \msn), \quad 
	\msv + \al \p_t \msv = \frac{\p W^*}{\p \msn}(\msm, \msn), \label{eq:implicitdiss}
\end{gather}
and
\begin{gather}
		\msm = \frac{\p W}{\p \msu}(\msu, \msv) + \mu \p_t \msu, \quad \msn = \frac{\p W}{\p \msv}(\msu, \msv) + \nu \p_t \msv, \label{eq:kelvinvoight}
\end{gather}
where $\al, \beta, \mu, \nu \geq 0$, $W^*$ is given by \eqref{eq:compenergy} and $W$ is given by \eqref{eq:stenergy}. The relations in \eqref{eq:implicitdiss} are strain-rate viscoelastic constitutive relations analogous to the small-strain strain-rate constitutive relations studied in \cite{RajSacc2014, BULPATSULISENGUL21, BULPATSULISENGUL22}. The relations \eqref{eq:kelvinvoight} are standard Kelvin-Voight type constitutive relations. In the authors' opinion, developing the stability theory of the equilibrium states discussed in Section 4, for a fixed boundary value problem and either choice of dissipation \eqref{eq:implicitdiss} or \eqref{eq:kelvinvoight}, is an interesting and worthwhile endeavor.  

\subsection{Modeling perspectives}

As mentioned in the introduction, the strain limiting constitutive relations introduced in this work would be ideal for modeling many rod-like materials whose final tangent stiffness greatly exceeds the initial tangent stiffness.  Such materials include collagen, elastin, silk, protein, DNA, RNA and many others.\footnote{Most materials such as DNA are viscoelastic, but as a first approximation can be considered as elastic bodies.}  De Gennes \cite{DeGennes1974} in describing such materials states ``the elongation tends to saturate: The restoring force F which tries to make a compact chain becomes infinite if it gets completely extended". That is, to get a compact chain to its full finite length, one needs infinite force, which we interpret in our idealization as limiting extensibility. As Freed and Rajagopal \cite{FreedRaj2016} observe, the idea that biological fibers can be strain-limiting was first propounded by Carton et al. \cite{CartonDainClark1962}. Later, Hunter \cite{Hunter1995} and Maksym and Bates \cite{MaksymBates97} also advocated the same notion.  

In particular, since DNA molecules are constantly bending, twisting and stretching inside cells during multiple biological processes, it is important to develop a simple model capturing their essential mechanical response. The small-strain special Cosserat rod model for double-stranded DNA (dsDNA) appearing in \cite{Lipfertetal14} suggests that for isolated tensile forces up to approximately 50 pN, our model with 
	\begin{gather}
		\gamma_{D}\al_D^2 = 185 \mbox{ pN (nm)}^2, \quad \gamma_D \beta_D^2 = 447 \mbox{ pN (nm)}^2, \\ 
		\gamma_D \eta_D^2 = 1000 \mbox{ pN}, \quad \gamma_D \iota_D = -70 \mbox{ pN  nm}, \label{eq:DNA}
	\end{gather} 
accurately describes the response of dsDNA (the rod is taken to be unshearable, so $v_1 = v_2 = 0$ always). However, as discovered by experiments on individual dsDNA molecules using optical tweezers \cite{Smith96}, for isolated tensile forces above approximately 65 pN, a dsDNA molecule undergoes a force-dependent, rate-dependent \emph{overstretching transition} wherein the molecule elongates to approximately 1.8 times its contour length and its response asymptotically approaches that of a single-stranded DNA molecule (ssDNA). 

\begin{figure}[b]
	\centering
	\includegraphics[scale=.5]{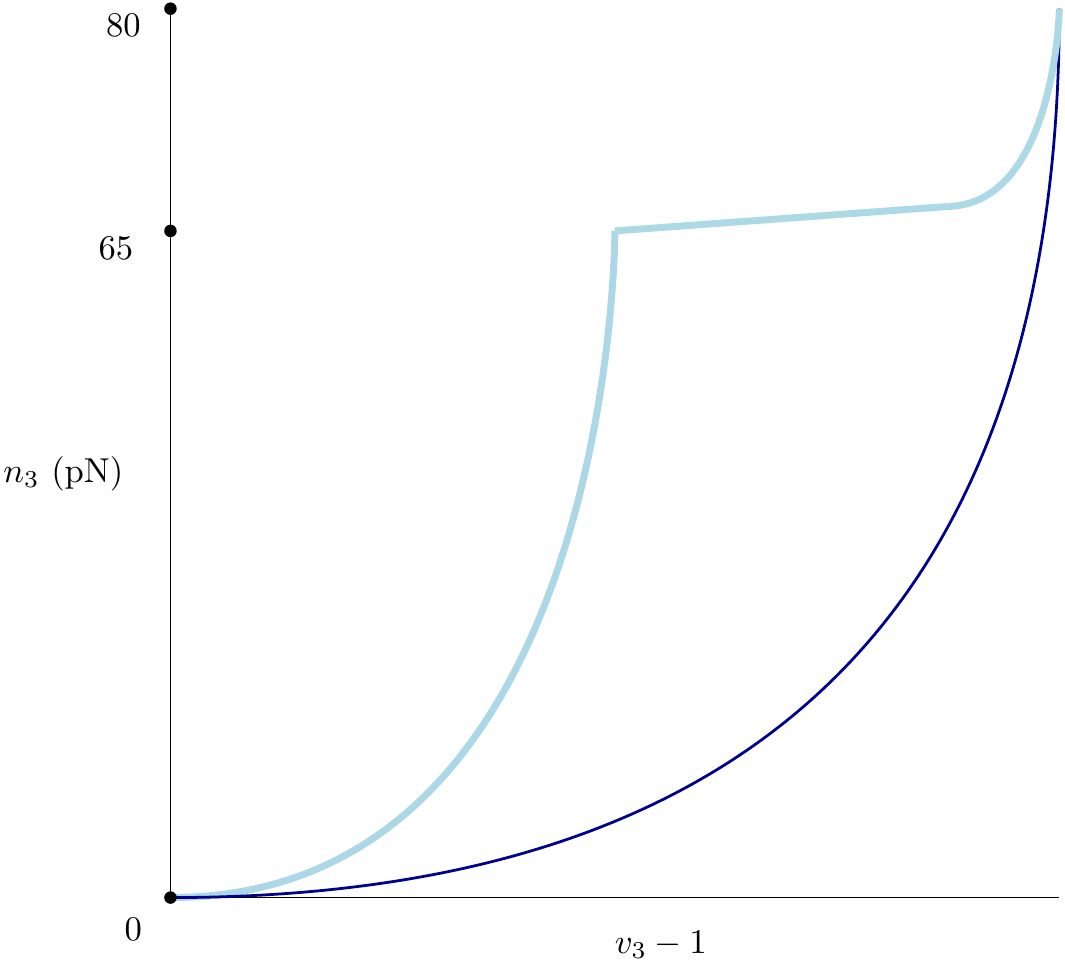}
	\caption{Schematic of the overstretching transition of dsDNA with the Carolina blue curve representing the force-stretch curve for dsDNA and the navy curve representing the force-stretch curve for ssDNA.}
	\label{f:fig4}
\end{figure}

After much debate over several years, its been revealed via experiments using both optical tweezers and fluorescent microscopy that there are three molecular mechanisms involved in the overstretching transition:
\begin{itemize}
	\item the peeling of one strand away from the other, 
	\item base-pair bonds melting, 
	\item base-pair bonds remaining intact and cooperative strand unwinding, converting parts of the molecule into ladder like structures (S-form DNA). 
\end{itemize} 
We refer the reader to the reviews \cite{ZaltronetalReview, BustamantePrimer} for more on the literature, experimental techniques and results leading to these conclusions. 

The precise superposition of strand splitting, bond melting and conversion into S-form DNA can be quite complex and difficult to track at the molecular level during the overstretching transition. However, using the special Cosserat rod model introduced in this work, the overstretching transition can be modeled by a rod with material constants $\al, \beta, \gamma, \iota, \eta$ and $p$ initially given by \eqref{eq:DNA} and converging to those of ssDNA, as an isolated tensile end thrust is applied at higher and higher forces. This process is rate-dependent and hysteretic \cite{Smith96}. The thermodynamic framework introduced by the first author and Srinivasa for  evolving natural configurations of three-dimensional bodies \cite{RAJAGOPALSRI2000, RajSriZAMP04a, RajSriZAMP04b} provides a road map for developing a thermodynamically consistent model capable of describing this phase transition at the level of continuum rods, a topic to be discussed in future work. 

\bibliographystyle{plain}
\bibliography{researchbibmech}
\bigskip

\centerline{\scshape K. R. Rajagopal}
\smallskip
{\footnotesize
	\centerline{Department of Mechanical Engineering, Texas A\&M University}
	
	\centerline{College Station, TX 77843, USA}
	
	\centerline{\email{krajagopal@tamu.edu}}
}

\bigskip

\centerline{\scshape C. Rodriguez}
\smallskip
{\footnotesize
	\centerline{Department of Mathematics, University of North Carolina}
	
	\centerline{Chapel Hill, NC 27599, USA}
	
	\centerline{\email{crodrig@email.unc.edu}}
}

\end{document}